\documentclass[a4paper,conference]{IEEEtran}

\usepackage{cite}

  \usepackage[dvips]{graphicx}
  \graphicspath{{figures/}}
\usepackage[cmex10]{amsmath}
\usepackage{algorithmicx}
\usepackage{algpseudocode}

\usepackage{tabularx}

  \usepackage[caption=false,font=footnotesize]{subfig}
\usepackage{url}

\usepackage{mathtools}
\usepackage{amssymb}
\usepackage{dsfont}
\usepackage{etoolbox}
\usepackage{accents}
\robustify{\underaccent}
\usepackage[colorlinks,hypertexnames=false]{hyperref}
\usepackage[capitalize]{cleveref}
\crefname{equation}{}{}
\usepackage{longtable}
\usepackage{bm}
\usepackage{amsthm}
\usepackage{commath}
\usepackage{framed}
\newtheorem{lemma}{Lemma}[section]
\newtheorem{proposition}{Proposition}[section]
\theoremstyle{definition}
\IEEEoverridecommandlockouts
\begin{document}
%
\title{Green OFDMA Resource Allocation in Cache-Enabled CRAN}

\author{Reuben~George~Stephen
~and~Rui~Zhang
\thanks{R. G. Stephen is with the NUS Graduate School for Integrative Sciences and Engineering~(NGS), National University of Singapore~(e-mail: reubenstephen@u.nus.edu). He is also with the Department of Electrical and Computer Engineering, National University of Singapore.}
\thanks{R. Zhang is with the Department of Electrical and Computer Engineering, National University of Singapore~(e-mail: elezhang@nus.edu.sg). He is also with the Institute for Infocomm Research, A*STAR, Singapore.}
\thanks{\copyright 2016 IEEE. Personal use of this material is permitted. Permission from IEEE must be obtained for all other uses, in any current or future media, including reprinting/republishing this material for advertising or promotional purposes, creating new collective works, for resale or redistribution to servers or lists, or reuse of any copyrighted component of this work in other works.}
}

\maketitle

\begin{abstract}
Cloud radio access network~(CRAN), in which remote radio heads~(RRHs) 
are deployed to serve users in a target area, and connected to a central processor~(CP) via limited-capacity links termed the fronthaul, is a promising candidate for the next-generation wireless communication systems. Due to the content-centric nature of future wireless communications, it is desirable to cache popular contents beforehand at the RRHs, to reduce the burden on the fronthaul and achieve energy saving through cooperative transmission. This motivates our study in this paper on the energy efficient transmission in an orthogonal frequency division multiple access~(OFDMA)-based CRAN with multiple RRHs and users, where the RRHs can prefetch popular contents. We consider a joint optimization of the user-SC assignment, RRH selection and transmit power allocation over all the SCs to minimize the total transmit power of the RRHs, subject to the RRHs' individual fronthaul capacity constraints and the users' minimum rate constraints, while taking into account the caching status at the RRHs. Although the problem is 
non-convex, we propose a Lagrange duality based solution, which can be efficiently computed with good accuracy. We compare the minimum transmit power required by the proposed algorithm with different caching strategies against the case without caching by simulations, which show the significant energy saving with caching. 
\end{abstract}
\section{Introduction}
Cloud radio access network~(CRAN) provides a cost-effective way to achieve network densification and hence meet the exponential growth in wireless network traffic, by replacing the conventional base stations~(BSs) with low-power 
distributed remote radio heads~(RRHs) that are 
coordinated by a central processor~(CP)
~\cite{bi-etal2015wireless}. In addition, 
CRAN offers both improved spectral efficiency and energy efficiency compared to conventional cellular networks, due to the centralized resource allocation and joint signal processing over the RRHs at the CP~\cite{shi-etal2014group,luo-etal2015downlink,liu-etal2015joint,stephen-zhang2016joint,dai-yu2016energy,stephen-zhang2016fronthaul}. However, along with the growth in the amount of wireless data traffic, the type of services required by users is also making a transition from the traditional \emph{connection-centric} communications such as voice calls and web surfing, to the so-called \emph{content-centric} communications such as video streaming, mobile application downloads, etc.~\cite{bi-etal2015wireless,tao-etal2016content}. An important characteristic of such content-centric communication is that the same contents are requested by multiple users at similar time. In order to address this paradigm shift in the nature of wireless traffic, it has been proposed to employ \emph{cache-enabled} RRHs in a CRAN~\cite{bi-etal2015wireless,tao-etal2016content}, where the RRHs can store popular contents beforehand, and hence, transmit the data requested by the users directly, without the need of fetching it from the CP over the fronthaul. Such a network architecture is also referred to as Fog Radio Access Network~(F-RAN)~\cite{park-etal2016joint}. Moreover, if the popular contents are cached at many RRHs in a CRAN, all of them can cooperatively transmit the data to many users at the same time, offering additional beamforming gains~\cite{liu-lau2013mixed,peng-etal2014joint,tao-etal2016content,park-etal2016joint}, and hence reducing the total transmit power required to satisfy the users' content requests. 
For a single-channel wireless system, the joint caching and transmit beamforming design was considered in~\cite{liu-lau2013mixed,peng-etal2014joint}. 
When the caching placement is known, a joint optimization of the BS clustering and transmit beamforming was studied in~\cite{tao-etal2016content}, while~\cite{park-etal2016joint} considered the joint optimization of transmit precoding and quantization noise covariances. 
In contrast to the above work, we consider the orthogonal frequency division multiple access~(OFDMA)-based CRAN with multiple sub-channels~(SCs), where the RRHs are enabled with caches of fixed size. Since the caching at the RRHs takes place over a larger time-scale compared to the wireless resource allocation, the cache status at the RRHs remains unchanged over many scheduling intervals, and hence it is assumed to be known for the resource allocation problem in this work, as in~\cite{tao-etal2016content,park-etal2016joint}. 
\begin{figure}[t]
\centering
\includegraphics[width=\linewidth]{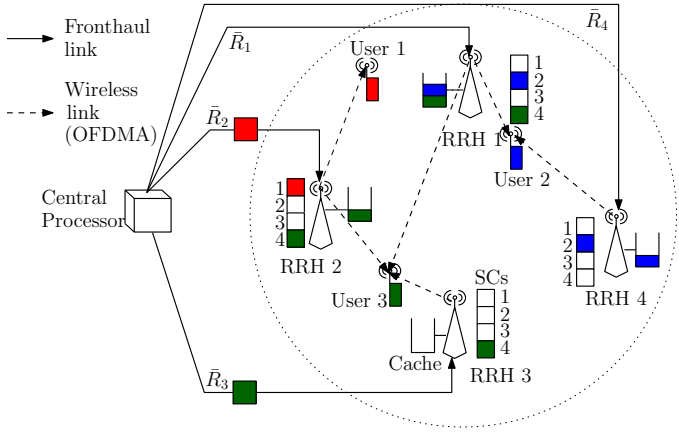}
\caption{Downlink of OFDMA-based CRAN with cache-enabled RRHs.}\label{F:SysModel}
\end{figure}

The availability of cached contents at multiple RRHs enables their cooperative transmission, thereby leading to energy savings in the network. For example, as shown in~\cref{F:SysModel}, due to the availability of user~3's contents at RRHs~1 and~2, the RRHs~1,~2~and~3 can cooperatively transmit to user~3, even though user~3 is located farther from RRHs~1 and~2 compared to RRH~3. 
Thus, with cache enabled RRHs, the user assignment, RRH selection and transmit power allocation on each SC must take into account the user requests and caching status at the RRHs, in addition to the individual fronthaul capacity constraints at the RRHs and minimum rate requirements at the users. Towards this end, in this paper we formulate a joint resource allocation problem in an OFDMA-based cache-enabled CRAN to minimize the total transmit power of all RRHs subject to to the users' minimum rate constraints in the downlink transmission. Although the problem is non-convex, we propose a Lagrange duality based algorithm, which is asymptotically optimal and  also efficient to implement. 
\section{System Model}
Consider the downlink of an OFDMA-based CRAN over bandwidth $B$~Hz, with $M$ single antenna RRHs denoted by $\mathcal{M}=\{1,\dotsc,M\}$, $K$ single-antenna users denoted by $\mathcal{K}=\{1,\dotsc,K\}$, $N$ SCs denoted by $\mathcal{N}=\{1,\dotsc,N\}$, and $F$ contents given by $\mathcal{F}=\{1,\dotsc,F\}$, as shown in~\cref{F:SysModel}. Each RRH can store up to $S\leq F$ contents in its cache. Let $c_{m,f}=1$ if content $f$ is cached at RRH $m$, and $c_{m,f}=0$ otherwise. Similarly let $u_{k,f}=1$ if user $k$ requests content $f$, and $u_{k,f}=0$ otherwise. Both the caching profile at the RRHs given by 
$\left\{c_{m,f}\right\}$, and the users' requests given by $\left\{u_{k,f}\right\}$ are assumed to be known at the CP. Each user requests at most one content at a time, i.e. $\sum_{f=1}^Fu_{k,f}=1,~\forall k\in\mathcal{K}$, but the same content can be requested by multiple users. Let $f_k\in\mathcal{F}$ denote the content requested by user $k\in\mathcal{K}$. Note that if a user does not request any content, then that user can be removed from consideration. The contents are transmitted to the users via OFDMA over one or more scheduling intervals depending on their size, while we consider the resource allocation optimization for a single scheduling interval. Let $\nu_{k,n}$ indicate whether user $k$ is assigned to SC $n$, i.e., 
\begin{align}
\nu_{k,n}=\begin{cases}1&\text{if user }k\text{ is assigned to SC }n\nonumber\\
0&\text{otherwise.}\end{cases}
\end{align}
Also define $\bm\nu_n\triangleq\begin{bmatrix}\nu_{1,n}&\cdots&\nu_{K,n}\end{bmatrix}^\mathsf{T}\in\{0,1\}^{K\times 1}$ as the user assignment on SC $n$. According to OFDMA, each SC $n\in\mathcal{N}$ is assigned to at most one user in the downlink transmission, and thus, $\bm 1^\mathsf{T}\bm\nu_n\leq 1,\enspace\forall n\in\mathcal{N}$. The set of SCs assigned to user $k$, denoted by $\mathcal{N}_k\subseteq\mathcal{N}$, is thus given by $\mathcal{N}_k=\left\{n\middle|\nu_{k,n}=1\right\}$, where $\mathcal{N}_j\cap\mathcal{N}_k=\emptyset,~\forall j\neq k,~j,k\in\mathcal{K}$. Since the fronthaul capacity for each RRH is practically limited, in general it can only receive the non-cached data for a selected subset of the users from the CP over its fronthaul, and then forward them to the selected users in the OFDMA-based downlink transmission. 
Let 
\begin{align}
\alpha_{m,n}=\begin{cases}
1&\text{if RRH }m\text{ transmits on SC }n\\
0&\text{otherwise.}
\end{cases}
\end{align}
Then, the subset of RRHs that transmit on SC $n$ is given by 
$\mathcal{A}_n=\left\{m\in\mathcal{M}\middle|\alpha_{m,n}=1\right\},\quad n\in\mathcal{N}$. 
Thus, the RRHs in $\mathcal{A}_n$ cooperatively send the data to the user $k$ assigned to SC $n$, i.e., $\nu_{k,n}=1$. Let $h_{k,m,n}$ 
denote the complex wireless access channel coefficient to the user $k\in\mathcal{K}$, from RRH $m\in\mathcal{M}$, on SC $n\in\mathcal{N}$, which is known at the CP, and $p_{m,n}\geq 0$ denote the power allocated by RRH $m$ on SC $n$. 
Then, with coherent transmission by all the RRHs in $\mathcal{A}_n$, the SNR at the receiver of the user $k$ assigned to SC $n$ can be expressed as~\cite{stephen-zhang2016joint}\footnote{For a complex scalar $x$, $\left|x\right|$ denotes its magnitude. 
}
\begin{align}
\gamma_{k,n}\left(\bm\alpha_n,\bm p_n\right)
=\frac{1}{\sigma^2}\left(\sum_{m=1}^M\left|h_{k,m,n}\right|\alpha_{m,n}\sqrt{p_{m,n}}\right)^2,\label{E:SCSNR}
\end{align}
$n\in\mathcal{N}_k$, where 
$\sigma^2$ is the 
power of the additive white Gaussian noise~(AWGN) at the receiver, which is assumed to be equal at all users. 
The achievable rate on SC $n\in\mathcal{N}_k$ is thus 
\begin{align}
r_{k,n}\left(\bm\alpha_n,\bm p_n\right)=\frac{B}{N}\log_2\left(1+\gamma_{k,n}\left(\bm\alpha_n,\bm p_n\right)\right).\label{E:SCRate}
\end{align}
Next, we present the following result on the concavity of the function $r_{k,n}\left(\bm\alpha_n,\bm p_n\right)$. 
\begin{lemma}\label{L:rknConc}
With given RRH selection $\bm\alpha_n$, $r_{k,n}\left(\bm\alpha_n,\bm p_n\right)$ 
defined in~\eqref{E:SCRate} is 
jointly concave with respect to $\left\{p_{m,n}\right\}$, $\forall m$ with $\alpha_{m,n}=1$. 
\end{lemma}
\begin{proof}
Please refer to~\cite[Appendix~A]{stephen-zhang2016joint}.
\end{proof}
If, on any SC $n$, RRH $m$ transmits to a user whose requested content is not cached, then, this user's data must be transmitted to RRH $m$ over its fronthaul link, from the CP. However, depending on the popularity profile, since several users may request the same content, each RRH only needs to fetch the unique data corresponding to a particular content from the CP. Moreover, this transfer of data from the CP to the RRH must be at a rate that is at least equal to the maximum rate at which it is to be transmitted over the OFDMA SCs, where the maximization is over all the users requesting this content. Thus, the rate at which RRH $m$ must receive over its fronthaul, the unique data 
to be transmitted to the users over all $N$ SCs, must not exceed its fronthaul capacity $\bar{R}_m$, which is expressed by the constraint 
\begin{align}
&\sum_{f=1}^F\left(1-c_{m,f}\right)\max_{k\in\mathcal{K}}\left\{u_{k,f}\sum_{n=1}^N\alpha_{m,n}\nu_{k,n}r_{k,n}\left(\bm\alpha_n,\bm p_n\right)\right\}\leq\bar{R}_m,\notag\\&~\forall m\in\mathcal{M}.
\label{C:FHRC}
\end{align}
Note that if a particular content $f$ is not requested by any user, i.e. $u_{k,f}=0~\forall k\in\mathcal{K}$, or if the content is already cached at the RRH, i.e. $c_{m,f}=1$, it does not contribute to the 
summation over $f$ in~\eqref{C:FHRC}. In the next section, we formulate the proposed joint resource allocation problem. 
\section{Problem Formulation}
We aim to minimize the total transmit power of the RRHs over all the SCs subject to the minimum rate constraints at each user denoted by $\underaccent{\bar}{R}_k,~k\in\mathcal{K}$, and the fronthaul rate constraints at the RRHs, denoted by $\bar{R}_m,~m\in\mathcal{M}$, by optimizing the user-SC assignments $\left\{\bm\nu_n\right\}_{n\in\mathcal{N}}$, RRH selections $\left\{\bm\alpha_n\right\}_{n\in\mathcal{N}}$ and the transmit power allocations by the RRHs $\left\{\bm p_n\right\}_{n\in\mathcal{N}}$ over the SCs. 
The problem can then be formally stated as below.
\begin{subequations}
\label{P:Main}
\begin{align}
\mathop{\mathrm{minimize}}_{\left\{\bm p_n,\bm\alpha_n,\bm\nu_n\right\}_{n\in\mathcal{N}}}&\enspace\sum_{m=1}^M\sum_{n=1}^N p_{m,n}\tag{\ref{P:Main}}\\
\mathrm{subject}~\mathrm{to}\notag\\
&\enspace\text{\cref{C:FHRC}}\notag\\
&\enspace\sum_{n=1}^N\nu_{k,n}r_{k,n}\left(\bm\alpha_n,\bm p_n\right)\geq\underaccent{\bar}{R}_k\quad\forall k\in\mathcal{K}\label{C:UserMinRateMain} \\
&\enspace p_{m,n}\geq 0\quad\forall m\in\mathcal{M},\enspace\forall n\in\mathcal{N}\label{C:PowerMain}\\
&\enspace\alpha_{m,n}\in\{0,1\}\quad\forall m\in\mathcal{M},\enspace\forall n\in\mathcal{N}\label{C:RRHSelMain}\\
&\enspace\bm 1^\mathsf{T}\bm\nu_n\leq 1\quad\forall n\in\mathcal{N}\label{C:USCSUMain}\\
&\enspace\nu_{k,n}\in\{0,1\}\quad\forall k\in\mathcal{K},\enspace\forall n\in\mathcal{N}.\label{C:SCAMain}
\end{align}
\end{subequations}
If there is no caching of contents at the RRHs, it is generally more energy efficient to have the nearest RRHs to a user to cooperatively transmit to it on any SC. However, with a given caching and user request profile, the RRH selection and user assignment on each SC also need to take into account the availability of cached contents at the RRHs, to maximize their cooperative transmission gain and also reduce their fronthaul rates required. 

Next, to simplify constraint~\eqref{C:FHRC}, we introduce auxiliary variables $\rho_{m,f}\geq 0$ such that $\rho_{m,f}=\tfrac{1}{\bar{R}_m}\max_{k\in\mathcal{K}}\left\{u_{k,f}\sum_{n=1}^N\alpha_{m,n}\nu_{k,n}r_{k,n}\left(\bm\alpha_n,\bm p_n\right)\right\}$. 
Then, problem~\eqref{P:Main} can be equivalently expressed as follows,
\begin{subequations}
\label{P:MainAV}
\begin{align}
\mathop{\min}_{\substack{\left\{\bm p_n,\bm\alpha_n,\bm\nu_n\right\}_{n\in\mathcal{N}},\\
\left\{\rho_{m,f}\right\}
}}&\enspace\enspace\sum_{m=1}^M\sum_{n=1}^N p_{m,n}\tag{\ref{P:MainAV}}\\
\mathrm{s.t.}\nonumber\\
&\enspace 
\frac{u_{k,f}\left(1-c_{m,f}\right)}{\bar{R}_m}\sum_{n=1}^N\alpha_{m,n}\nu_{k,n}r_{k,n}\left(\bm\alpha_n,\bm p_n\right)\notag\\
&\enspace\leq u_{k,f}\left(1-c_{m,f}\right)\rho_{m,f}\notag\\
&\enspace
\forall m\in\mathcal{M},\forall f\in\mathcal{F},\forall k\in\mathcal{K}
\label{C:MaxFHRUCMainAV}\\
&\enspace\sum_{f=1}^F\left(1-c_{m,f}\right)\rho_{m,f}\leq 1
\quad\forall m\in\mathcal{M}\label{C:FHRCMainAV}\\
&\enspace\text{\cref{C:UserMinRateMain,C:PowerMain,C:RRHSelMain,C:USCSUMain,C:SCAMain}.}\notag
\end{align}
\end{subequations}
Note that  the constraints in~\eqref{C:MaxFHRUCMainAV} and~\eqref{C:FHRCMainAV} do not apply for those RRHs that have already cached a particular content requested by some user. Thus, the number of constraints in~\eqref{C:MaxFHRUCMainAV} and~\eqref{C:FHRCMainAV} depend on the caching status at the RRHs and the content requests by the users. Problem~\eqref{P:MainAV} is non-convex due to the integer constraints on the RRH selections $\bm\alpha_n$ and user-SC assignments $\bm\nu_n,~n\in\mathcal{N}$. Even if $\bm\alpha_n$ and $\bm\nu_n$ are fixed on all $n$, problem~\eqref{P:MainAV} is still non-convex, since constraint~\eqref{C:MaxFHRUCMainAV} is non-convex due to~\cref{L:rknConc}. 
\section{Proposed Solution}
In~\cite{yu-lui2006dual}, 
it was shown that the duality gap for non-convex optimization problems, where the objective and constraints are separable over the SCs, goes to zero as the number of SCs goes to infinity. 
In problem~\eqref{P:MainAV}, the objective, as well as the constraints~\eqref{C:MaxFHRUCMainAV} are separable over the SCs, while the constraints~\eqref{C:FHRCMainAV}, although not separable over the SCs, are convex, and hence do not affect the convexity of the problem. Thus, due to the time-sharing property of~\cite{yu-lui2006dual}, the duality gap of problem~\eqref{P:MainAV} also goes to zero as $N$ goes to infinity. Since $N$ is typically large in practice, we thus propose to apply the Lagrange duality method to solve problem~\eqref{P:MainAV}. 
Let $\lambda_{m,k,f}\geq 0,~f\in\mathcal{F},~k\in\mathcal{K},~m\in
\mathcal{M}
$ denote the dual variables associated with the fronthaul constraints in~\eqref{C:MaxFHRUCMainAV}, and $\mu_k\geq 0,~k\in\mathcal{K}$, denote the dual variables corresponding to the 
minimum rate constraints at the users in~\eqref{C:UserMinRateMain}. If $\mathcal{M}^\textsf{c}_{f_k}\triangleq\left\{m|c_{m,f_k}=0\right\},~k\in\mathcal{K}$ denotes the set of RRHs that do not have the content $f_k$ requested by user $k\in\mathcal{K}$, then 
there are $C\triangleq\sum_{k=1}^K\left|\mathcal{M}^\textsf{c}_{f_k}\right|$ constraints in~\eqref{C:MaxFHRUCMainAV}.\footnote{For a finite set $\mathcal{A}$, $\left|\mathcal{A}\right|$ denotes its cardinality.} Let the vectors $\bm\lambda
\in\mathds{R}_+^{C\times 1}$ and $\bm\mu\in\mathds{R}_+^{K\times 1}$, denote the collections of  these dual variables. Then, the (partial) Lagrangian of problem~\eqref{P:MainAV} with respect to the constraints in~\eqref{C:MaxFHRUCMainAV} and~\eqref{C:UserMinRateMain} can be expressed as 
\begin{align}
&L\left(\left\{\bm\nu_n,\bm\alpha_n,\bm p_n\right\}_{n\in\mathcal{N}},\left\{\rho_{m,f}\right\}
,\bm\lambda,\bm \mu\right)\notag\\
&=\sum_{n=1}^N L_n\left(\bm\nu_n,\bm\alpha_n,\bm p_n,\bm\lambda,\bm \mu\right)\notag\\
&\quad-\sum_{m=1}^M\sum_{k=1}^K\sum_{f=1}^F\lambda_{m,k,f}u_{k,f}\left(1-c_{m,f}\right)\rho_{m,f}+\sum_{k=1}^K\mu_k,
\label{E:LagW}
\end{align}
where 
\begin{align}
&L_n\left(\bm\nu_n,\bm\alpha_n,\bm p_n,\bm \lambda,\bm \mu\right)\notag\\
&\triangleq\sum_{m=1}^M p_{m,n}+\sum_{m=1}^M\sum_{k=1}^K\sum_{f=1}^F\alpha_{m,n}\nu_{k,n}u_{k,f}\left(1-c_{m,f}\right)\frac{\lambda_{m,k,f}}{\bar{R}_m}\notag\\
&\enspace\cdot r_{k,n}\left(\bm\alpha_n,\bm p_n\right)-\sum_{k=1}^K\frac{\mu_k}{\underaccent{\bar}{R}_k}\nu_{k,n}r_{k,n}\left(\bm\alpha_n,\bm p_n\right).\label{E:Lagn}
\end{align}
The Lagrange dual function is thus given by  
\begin{subequations}
\label{E:DualFunc} 
\begin{align}
&g(\bm\lambda,\bm\mu)=\notag\\
&~\min_{\substack{\left\{\bm p_n,\bm\alpha_n,\bm\nu_n\right\}_{n\in\mathcal{N}}\\\left\{\rho_{m,f}\right\}
}}~L\left(\left\{\bm\nu_n,\bm\alpha_n,\bm p_n,\right\}_{n\in\mathcal{N}},\left\{\rho_{m,f}\right\}
,\bm\lambda,\bm\mu\right)\tag{\ref{E:DualFunc}}\\
&~\mathrm{s.t.}
~\text{\cref{C:FHRCMainAV,C:PowerMain,C:RRHSelMain,C:USCSUMain,C:SCAMain}}.\notag
\end{align}
\end{subequations}
Using~\eqref{E:LagW}, the dual function in~\eqref{E:DualFunc} can be 
expressed as 
\begin{align}
g(\bm\lambda,\bm\mu)=g_1(\bm\lambda,\bm\mu)+g_2(\bm\lambda,\bm\mu)+\sum_{k=1}^K\mu_k,\label{E:DualFuncSplit}
\end{align}
where
\begin{subequations}
\label{P:DualFunc1} 
\begin{align}
g_1(\bm\lambda,\bm\mu)=~\min_{\left\{\bm p_n,\bm\alpha_n,\bm\nu_n\right\}_{n\in\mathcal{N}}}&\enspace \sum_{n=1}^N L_n\left(\bm\nu_n,\bm\alpha_n,\bm p_n,\bm\lambda,\bm \mu\right)\tag{\ref{P:DualFunc1}}\\
\mathrm{s.t.}
&\enspace\text{\cref{C:PowerMain,C:RRHSelMain,C:USCSUMain,C:SCAMain}}.\notag
\end{align}
\end{subequations}
and
\begin{subequations}
\label{P:DualFunc2} 
\begin{align}
g_2(\bm\lambda,\bm\mu)&=
\min_{\left\{\rho_{m,f}\right\}
}~-\sum_{m=1}^M\sum_{f=1}^F\sum_{k=1}^Ku_{k,f}\left(1-c_{m,f}\right)
\lambda_{m,k,f}\rho_{m,f}
\tag{\ref{P:DualFunc2}}\\
&\mathrm{s.t.}~\text{\cref{C:FHRCMainAV}}.\notag
\end{align}
\end{subequations}
The minimization problem in~\eqref{P:DualFunc1} can be decomposed into $N$ parallel sub-problems, where each sub-problem corresponds to a single SC $n\in\mathcal{N}$, and all of them have the same structure given by
\begin{subequations} 
\label{P:DualFunc1n}
\begin{align}
\min_{\bm p_n,\bm\alpha_n,\bm\nu_n}&\enspace
L_n\left(\bm\nu_n,\bm\alpha_n,\bm p_n,\bm\lambda,\bm \mu\right)\tag{\ref{P:DualFunc1n}}\nonumber\\
\mathrm{s.t.}
&\enspace\bm p_n\succeq\bm 0\label{C:VectPNZDFn}\\
&\enspace\bm\alpha_n\in\{0,1\}^{M\times 1}\label{C:BinRRHSelVectDFn}\\
&\enspace\bm 1^\mathsf{T}\bm\nu_n\leq 1\\
&\enspace\bm\nu_n\in\{0,1\}^{K\times 1}
\end{align}
\end{subequations}
where $L_n\left(\bm\nu_n,\bm\alpha_n,\bm p_n,\bm\lambda,\bm \mu\right)$ is defined in~\eqref{E:Lagn}. Next, we describe how to solve problem~\eqref{P:DualFunc1n} on each SC $n$. 

Let the user association on SC $n$ be fixed as $\bm\nu_n=\hat{\bm\nu}_n$. If $\hat{\bm\nu}_n=\bm 0$, no user is assigned to SC $n$, and since $\bm p_n\succeq\bm 0$, the objective of problem~\cref{P:DualFunc1n} as given in~\cref{E:Lagn}, is minimized by setting $\bm p_n=\bm 0$, irrespective of the RRH selection $\bm\alpha_n$. Thus, if no user is assigned to SC $n$, the power allocation over all RRHs is zero, as expected, and we assume $\bm\alpha_n=\bm 0$ without loss of generality. Otherwise, let $\hat{k}_n\in\mathcal{K}$ be the user assigned to SC $n$ so that $\hat{\nu}_{\hat{k}_n,n}=1$ and $\hat{\nu}_{k,n}=0\enspace\forall k\neq\hat{k}_n$. Also, let $f_{\hat{k}_n}\in\mathcal{F}$ denote the content requested by this user $\hat{k}_n$, where $u_{\hat{k}_n,f_{\hat{k}_n}}=1$ and $u_{\hat{k}_n,f}=0,~\forall f\neq f_{\hat{k}_n}$. Then, 
problem~\eqref{P:DualFunc1n} on each SC $n$ is reduced to 
\begin{subequations}
\label{P:DualFunc1nFixUA}
\begin{align}
\min_{\bm p_n,\bm\alpha_n}&\enspace-\left(\frac{\mu_{\hat{k}_n}}{\underaccent{\bar}{R}_{\hat{k}_n}}-\sum_{m=1}^M\left(1-c_{m,f_{\hat{k}_n}}\right)\alpha_{m,n}\frac{\lambda_{m,\hat{k}_n,f_{\hat{k}_n}}}{\bar{R}_m}\right)\notag\\
&\enspace\cdot r_{\hat{k}_n,n}\left(\bm\alpha_n,\bm p_n\right)+\sum_{m=1}^Mp_{m,n}\tag{\ref{P:DualFunc1nFixUA}}\\
\mathrm{s.t.}&\enspace\text{\cref{C:VectPNZDFn,C:BinRRHSelVectDFn}},\notag
\end{align}
\end{subequations}
which is non-convex due to the integer constraints~\eqref{C:BinRRHSelVectDFn} on $\bm\alpha_n$ and the coupled variables in the objective. However, for a given RRH selection $\tilde{\bm\alpha}_n$, the optimal power allocation $\tilde{\bm p}_n$ that solves problem~\eqref{P:DualFunc1nFixUA} is given by the following proposition. 
\begin{proposition}\label{Prop:OptPAFixURRHSel}
Let $\bm\alpha_n=\tilde{\bm\alpha}_n$ be fixed. Then, 
the optimal 
power allocation $\tilde{\bm p}_n$ on SC $n$ for 
problem~\eqref{P:DualFunc1nFixUA} is given by
\begin{align}
\tilde{p}_{m,n}&=\left[\frac{B\cdot F_{\hat{k}_n,n}\left(\tilde{\bm\alpha}_n\right)G_{\hat{k}_n,n}\left(\tilde{\bm\alpha}_n\right)}{N\ln 2}-1\right]^+\frac{\tilde{\alpha}_{m,n}\left|h_{\hat{k}_n,m,n}\right|^2}{\left(G_{\hat{k}_n,n}\left(\tilde{\bm\alpha}_n\right)\right)^2\sigma^2}\notag\\
\label{E:OptPA}
\end{align}
$\forall m\in\mathcal{M}$, where 
\begin{align}
F_{\hat{k}_n,n}\left(\bm\alpha_n\right)&\triangleq~\frac{\mu_{\hat{k}_n}}{\underaccent{\bar}{R}_{\hat{k}_n}}-\sum_{m=1}^M\frac{\left(1-c_{m,f_{\hat{k}_n}}\right)\alpha_{m,n}\lambda_{m,\hat{k}_n,f_{\hat{k}_n}}}{\bar{R}_m}\label{E:FHParam}\\
G_{\hat{k}_n,n}\left(\bm\alpha_n\right)
&\triangleq~\sum_{m=1}^M\frac{\alpha_{m,n}\left|h_{\hat{k}_n,m,n}\right|^2}{\sigma^2}.
\label{E:AccParam}
\end{align}
\end{proposition}
\begin{proof}
The proof is similar to that in~\cite[Appendix~B]{stephen-zhang2016joint}.
\end{proof}
\Cref{Prop:OptPAFixURRHSel} shows that for given user association $\hat{k}_n$ and RRH selection $\tilde{\bm\alpha}_n$ on each SC $n$, the optimal power allocation has a threshold structure, which allocates zero power to all RRHs on SC $n$ if $F_{\hat{k}_n,n}\left(\tilde{\bm\alpha}_n\right)G_{\hat{k}_n,n}\left(\tilde{\bm\alpha}_n\right)\leq(N\ln 2)/B$. 
Otherwise, if $F_{\hat{k}_n,n}\left(\tilde{\bm\alpha}_n\right)G_{\hat{k}_n,n}\left(\tilde{\bm\alpha}_n\right)>(N\ln 2)/B$, the power allocation on each RRH $m\in\mathcal{A}_n$ 
depends on the wireless access channel gain $\left|h_{\hat{k}_n,m,n}\right|$ on SC $n$. and the dual variable $\mu_{\hat{k}_n}$ corresponding to the rate constraint~\eqref{C:UserMinRateMain}. Also, if all the RRHs have cached the content $f_{\hat{k}_n}$ requested by the user, i.e., if $c_{m,f_{\hat{k}_n}}=1,~\forall m\in\mathcal{M}$, the power allocation reduces to
\begin{align}
\tilde{p}_{m,n}&=\left[\frac{B\mu_{\hat{k}_n}G_{\hat{k}_n,n}\left(\tilde{\bm\alpha}_n\right)}{\underaccent{\bar}{R}_{\hat{k}_n}N\ln 2}-1\right]^+\frac{\tilde{\alpha}_{m,n}\left|h_{\hat{k}_n,m,n}\right|^2}{\left(G_{\hat{k}_n,n}\left(\tilde{\bm\alpha}_n\right)\right)^2\sigma^2}
\end{align}
If $\tilde{\bm\alpha}_n=\bm 0$, i.e., no RRH is selected, then $\tilde{\bm p}_n=\bm 0$. 
For the special case when there is only one RRH in the cluster, the power allocation in~\eqref{E:OptPA} becomes~(drop the subscript $m$)
\begin{align}
\tilde{p}_n=\left[\frac{B}{N\ln 2}\left(\frac{\mu_{\hat{k}_n}}{\underaccent{\bar}{R}_{\hat{k}_n}}-\frac{\left(1-c_{f_{\hat{k}_n}}\right)\lambda_{\hat{k}_n,f_{\hat{k}_n}}}{\bar{R}}\right)-\frac{\sigma^2}{\left|h_{\hat{k}_n,n}\right|^2}\right]^+,\label{E:OptPASingleRRH}
\end{align} 
which has the same form as the well-known water-filling solution, but in general with different water levels on different SCs $n\in\mathcal{N}$. 
If $\left(\mu_{\hat{k}_n}/\underaccent{\bar}{R}_{\hat{k}_n}\right)\leq\left(1-c_{f_{\hat{k}_n}}\right)\lambda_{\hat{k}_n,f_{\hat{k}_n}}/\bar{R}$, no power should be allocated to SC $n$. Thus, for given dual variables $\bm\lambda$ and $\bm\mu$, problem~\eqref{P:DualFunc1n} can be solved optimally using~\cref{Prop:OptPAFixURRHSel} as follows. First, fix the user on SC $n$ as $\hat{k}_n\in\mathcal{K}$. Then, for each of the $2^M$ possible RRH selections, compute the optimal power allocation $\tilde{\bm p}_n$ using~\eqref{E:OptPA}, and choose the optimal RRH selection $\hat{\bm\alpha}_n$ for the user $\hat{k}_n$ as the one that maximizes the objective of problem~\eqref{P:DualFunc1nFixUA} with the corresponding power allocation $\hat{\bm p}_n$ given by~\eqref{E:OptPA}. Then the optimal user association $\bar{\bm\nu}_n$ on SC $n$ can be found by 
choosing the user $\bar{k}_n$ that maximizes the objective of problem~\eqref{P:DualFunc1n}, with its corresponding optimal RRH selection and power allocation computed before. 
Similarly, for given $\bm\lambda$ and $\bm\mu$, the optimal solution to the linear program~\eqref{P:DualFunc2} 
is given by the following proposition. 
\begin{proposition}\label{Prop:OptSolDF2}
The optimal solution to problem~\eqref{P:DualFunc2} is given by
\begin{align}
&\rho_{m,f}=\begin{cases}1&\text{if }f=\displaystyle{\mathop{\arg\max}_{\ell\in\mathcal{F}}\left\{\sum_{k=1}^Ku_{k,\ell}\left(1-c_{m,\ell}\right)\lambda_{m,k,\ell}\right\}}\\
0&\text{otherwise}
\end{cases},\notag\\&~m\in\mathcal{M}.
\end{align}
\end{proposition}
\begin{proof}
The proof follows by contradiction from the structure of problem~\eqref{P:DualFunc2}, and the details are omitted.
\end{proof}
\cref{Prop:OptSolDF2} shows that problem~\eqref{P:DualFunc2} can be solved by searching for the content $f$ with largest value of $\sum_{k=1}^Ku_{k,\ell}\left(1-c_{m,\ell}\right)\lambda_{m,k,\ell}$ among all the contents requested by at least one user, and not cached at each RRH $m$, and then setting $\rho_{m,f}=1$, while $\rho_{m,\ell}=0$ for the other contents $\ell\neq f$. 
The worst-case complexity of finding the value of $g_2\left(\bm\lambda,\bm\mu\right)$ is thus $O\left(MF\right)$, which is incurred when there are at least $F$ users, all the users request distinct contents, and none of those contents are cached by any of the RRHs. Now, the dual problem for~\eqref{P:Main} is given by 
\begin{align}
\max_{\bm\lambda\succeq 0,\bm\mu\succeq\bm 0}g(\bm\lambda,\bm\mu),\label{E:DualProb}
\end{align}
which is convex and can be solved efficiently, e.g., using the ellipsoid method 
to find the optimal dual variables $\bm\lambda^\star$ and $\bm\mu^\star$. Then, the optimal solutions to the problems~\cref{P:DualFunc1,P:DualFunc2} are given by $\left\{\bm\nu^\star_n,\bm\alpha^\star_n,\bm p^\star_n\right\}$ and $\left\{\rho^\star_{m,f}\right\}
$, computed as outlined above at the optimal dual variables $\bm\lambda^\star$ and $\bm\mu^\star$. The algorithm for solving problem~\eqref{P:Main} is thus given in~\cref{A:Overall}. 
\begin{table}[h]
\caption{Algorithm for problem~\eqref{P:Main}}\label{A:Overall}
\begin{framed}
\begin{algorithmic}[1]
\State Initialization: $\bm\lambda\succeq 0$, $\bm\mu\succeq\bm 0$
\Repeat 
\For {each $n\in\mathcal{N}$}
\State For each user $\hat{k}_n$ and RRH selection $\tilde{\bm\alpha}_n$, find optimal power allocation $\tilde{\bm p}_n$ using~\cref{Prop:OptPAFixURRHSel}\label{AL:ORRHSelP}
\State Choose RRH selection and corresponding optimal power allocation that minimizes objective in~\eqref{P:DualFunc1nFixUA} 
\State Choose user with optimal RRH selection and power allocation that minimizes objective in~\eqref{P:DualFunc1n} 
\EndFor
\State Solve problem~\eqref{P:DualFunc2} using~\cref{Prop:OptSolDF2} to get optimal $\left\{\rho_{m,f}\right\}
$
\State Update dual variables $\bm\lambda$ and $\bm\mu$ using the ellipsoid method
\Until ellipsoid algorithm converges to desired accuracy
\end{algorithmic}
\end{framed}
\end{table}

Finding the optimal RRH selection and power allocation for a given user association 
involves a search over $2^M$ values. 
Subsequently, finding the optimal user association involves a search over $K$ users. 
Each of the $N$ problems~\eqref{P:DualFunc1n} can thus be solved incurring a complexity of $O\left(K2^M\right)$ and hence, the computation of 
$g_1(\bm\lambda,\bm\mu)$ in~\eqref{E:DualFuncSplit} incurs an overall complexity of $O\left(NK2^M\right)$. Similarly, the worst-case complexity of solving problem~\eqref{P:DualFunc2} to compute $g_2(\bm\lambda,\bm\mu)$ in~\eqref{E:DualFuncSplit}, is $O\left(MF\right)$, according to~\cref{Prop:OptSolDF2}. The complexity of the ellipsoid method to find the optimal dual variables depends only on the size of the initial ellipsoid and the maximum length of the sub-gradients over the intial ellipsoid. 
Thus, the worst-case complexity of solving problem~\eqref{P:Main} using the algorithm in~\cref{A:Overall} is effectively given by $O\left(NK2^M+MF\right)$, which is not very high for reasonable cluster sizes with $M\leq 5$. Moreover, a greedy RRH selection as in~\cite{stephen-zhang2016joint,stephen-zhang2016fronthaul} may be used to further reduce the complexity to $O\left(NKM^2+MF\right)$. 
\section{Simulation Results}
For the simulation setup, we consider a CRAN cluster with $M=5$ RRHs. One RRH is located in the center of a square region with side $100$~meters~(m), while the others are located on the vertices. There are $K=10$ users randomly located within a larger square region with side $200$~m, whose center coincides with that of the RRH square region. The fronthaul links of all the RRHs are assumed to have the same capacity $\bar{R}_m=\bar{R}\enspace\forall m\in\mathcal{M}$. There are $F=50$ distinct contents and the users' requests follow a Zipf distribution~\cite{
peng-etal2014joint,guo-etal2015cooperative,tao-etal2016content,park-etal2016joint}, according to which the probability that a user requests content $f\in\mathcal{F}$ is given by $\pi_f =f^{-\eta}/\sum_{\ell=1}^F\ell^{-\eta},~ f\in\mathcal{F}$.
Here $\eta$ is a shaping parameter that determines the skew of the distribution, and is set as $\eta = 0.9$, which is a typical value~\cite{
guo-etal2015cooperative}, and assumed to be the same for all the users. The minimum data rate at which each user requests a content is $\underaccent{\bar}{R}_k=\underaccent{\bar}{R}=20$~Mbps $\forall k\in\mathcal{K}$, while the cache at each RRH is assumed to be capable of storing at most $S$ distinct contents.  

The wireless channel is centered at a frequency of $2$~GHz with a bandwidth $B=20$~MHz, following the Third Generation Partnership Project~(3GPP) Long Term Evolution-Advanced~(LTE-A) standard
, and is divided into $N=64$ SCs using OFDMA. The combined path loss and shadowing is modeled as $38+30\log_{10}\left(d_{k,m}\right)+X$ 
in dB
, where $d_{k,m}$ in~m is the distance between the RRH $m$ and the user $k$, and $X$ 
is the shadowing random variable, which is Gaussian distributed with a standard deviation of $6$~dB. The AWGN is assumed to have a power spectral density of $-174$~dBm/Hz with a noise figure of $9$~dB at each user. 
The multi-path on each wireless 
channel is modeled using an exponential power delay profile with $N/4$ taps 
and the small-scale fading on each tap is assumed to follow the Rayleigh distribution. 
\begin{figure}[t]
\centering
\includegraphics[scale=0.45]{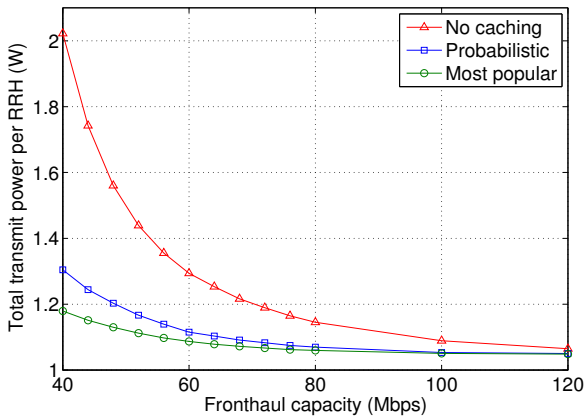}
\caption{Total transmit power normalized by number of RRHs $M$ vs.\ common fronthaul capacity $\bar{R}$ for system with $M=5$, $K=10$, $N=64$, $B=20$~MHz, $S=5$ and $\underaccent{\bar}{R}=20$~Mbps.}\label{F:SPvFHC}
\end{figure}
We compare the performance of the following three schemes:
\begin{itemize}
\item \textbf{Most popular content caching}: In this case, each RRH caches the most popular contents until its storage is full. Thus, in this case, if all the RRHs have the same storage size, all of them cache the same contents. With this caching status given, problem~\eqref{P:Main} is solved according to the algorithm in~\cref{A:Overall}.
\item \textbf{Probabilistic content caching}: In this case, each RRH independently caches contents according to their popularity probability until its storage is full. 
\item \textbf{No caching}: In this case, none of the RRHs cache any of the contents, and all the users' data need to be obtained from the CP by the RRHs, over the fronthaul links.
\end{itemize}
\cref{F:SPvFHC} plots the total transmit power required over all SCs normalized by the number of RRHs, in Watts~(W), against the common fronthaul link capacity $\bar{R}$, averaged over many random user locations and content request profiles. From~\cref{F:SPvFHC}, it is observed that caching popular contents at the RRHs leads to a significant savings in the transmit power compared to a system without caching, especially when the fronthaul capacity is low. Moreover, the deterministic most popular content caching is seen to perform better than the probabilistic caching according to popularity, since in this case there is maximum opportunity for cooperative transmission by the RRHs to most of the users, while such opportunities are in general less when the contents are independently cached by the RRHs. At larger fronthaul capacities, the transmit power saving offered by caching becomes less, since in this case, the fronthaul itself can support cooperative transmission by most of the RRHs, even without caching.
\begin{figure}[t]
\centering
\includegraphics[scale=0.45]{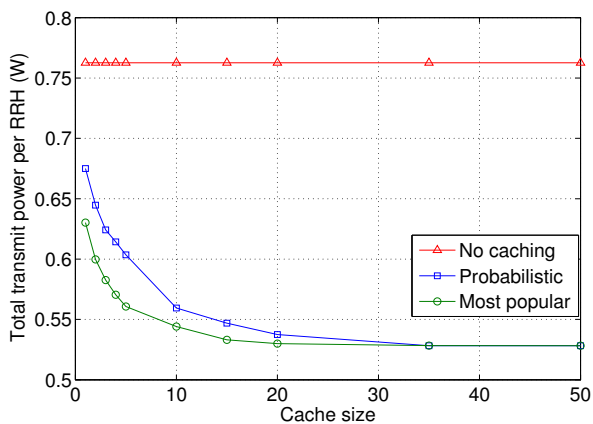}
\caption{Total transmit power normalized by number of RRHs $M$ vs.\ common cache size $S$ for system with $M=5$, $K=10$, $N=64$, $B=20$~MHz, $\bar{R}=80$~Mbps and $\underaccent{\bar}{R}=20$~Mbps.}\label{F:SPvCS}
\end{figure}
\cref{F:SPvCS} plots the transmit power against the common cache size $S$ at the RRHs. Similar trends as in~\cref{F:SPvCS} are observed, and caching even one file at each RRH can offer significant savings in the transmit power, leading to increased energy efficiency. 
\section{Conclusion}\label{Sec:Conc}
In this paper, we have studied the energy-efficient transmission design in an OFDMA-based CRAN with cache-enabled RRHs, when the caching status at the RRHs is known. We formulated a joint user association, RRH selection, and power allocation problem to minimize the total transmit power of the RRHs over all SCs subject to the RRHs' individual fronthaul capacity constraints and the minimum data rate constraints of the users. Although the problem is non-convex, we propose an efficient solution based on the Lagrange duality technique. Through numerical simulations, we compare different caching schemes, and show that the optimized resource allocation with caching, offers significant savings in transmit power  compared to a system with no caching at the RRHs, thus leading to a more energy-efficient network.
\bibliographystyle{IEEEtran_mod}
\bibliography{IEEEabrv,bibJournalList,ThesisBibliography}
%
%
%
\end{document}